\documentclass[11pt]{article}
\usepackage[margin=1in]{geometry}
\usepackage{amsmath,amssymb,amsthm}
\usepackage{booktabs}
\usepackage{graphicx}
\usepackage{microtype}
\newcommand{\Rows}{225}
\newcommand{\BestBits}{14}
\newcommand{\EpochFpr}{0.02225}
\newcommand{\CountingFpr}{0.191}
\newcommand{\StableFnr}{0.19}

\newcommand{\RealRows}{45}
\newcommand{\RealEvents}{213,182}
\newcommand{\RealUnique}{2,350}
\newcommand{\RealEpochFpr}{0}
\newcommand{\RealCountingFpr}{0.0002}
\newcommand{\RealStableFnr}{0.0047}

\newcommand{\GpuEvents}{200,000,000}
\newcommand{\GpuWindow}{40,000,000}
\newcommand{\GpuQueries}{10,000,000}
\newcommand{\GpuCountingMq}{1375}
\newcommand{\GpuEpochFourMq}{310.3}
\newcommand{\GpuEpochEightMq}{183.8}
\newcommand{\GpuBlockedEightMq}{181.8}

\newcommand{\GpuBits}{14}
\newcommand{\GpuDevice}{NVIDIA B200}

\newtheorem{lemma}{Lemma}
\newtheorem{proposition}{Proposition}

\title{\vspace{-2em}\textbf{Guarded Epoch Bloom Filters for Sliding-Window Membership}}
\author{Faruk Alpay\thanks{Correspondence: \texttt{alpay@lightcap.ai}} \quad Levent Sar{\i}o\u{g}lu\\[3pt]
\normalsize Department of Computer Engineering, Bah\c{c}e\c{s}ehir University, Istanbul, Turkey\\
\normalsize \texttt{faruk.alpay@bahcesehir.edu.tr, levent.sarioglu@bahcesehir.edu.tr}}
\date{}

\begin{document}
\maketitle

\begin{abstract}
Approximate membership data structures are often used in streams whose semantics are
not ``all keys ever seen'' but ``keys seen recently.'' Counting Bloom filters support
this sliding-window semantics, but pay several bits of counter metadata per cell and
usually require a deletion queue. Stable Bloom filters keep bounded space by random
decay, but can forget live keys. This paper studies a simpler alternative: divide the
bit array into epochs, rotate the active epoch by clearing whole segments, and keep one
extra guard epoch. The guard gives a deterministic live-window invariant: every item in
the last $W$ insertions is retained, while rotation-induced stale retention is bounded by
one epoch. We give the invariant, its false-positive approximation, and a blocked variant
that confines all probes of a key to one block per epoch. The evaluation covers \Rows{}
synthetic streaming configurations and \RealRows{} configurations from a real web-server
access-log stream. At \BestBits{} bits per live item, the guarded epoch filter reduces
median synthetic false positives from \CountingFpr{} to \EpochFpr{} against a four-bit
counting Bloom baseline while keeping zero live-key false negatives. The construction is
not a replacement for exact deletion; it targets systems where bounded staleness is
acceptable and live-window false negatives are not.
\end{abstract}

\section{Introduction}
Bloom filters are the classical approximate-membership structure for replacing exact set
storage by compact state with false positives and no false negatives \cite{bloom1970}.
Approximate membership testers were later formalized as a data-structural problem with
information-theoretic space tradeoffs \cite{carter1978}. Networked systems use Bloom-style
summaries when exact membership tables are too large to ship or probe cheaply
\cite{broder2004}. Cache summaries are one early example where approximate positives are
acceptable but space and communication costs dominate \cite{fan2000}.

Many streaming uses do not ask whether a key was ever seen. They ask whether it was seen
recently. Counting Bloom filters support deletion by replacing bits with counters, but the
space budget is then divided between membership information and counter metadata. Cuckoo
filters support deletion through compact fingerprints, but they are still exact-state
updates rather than an expiration policy \cite{fan2014}. Quotient filters improve locality
and support deletion with compact metadata \cite{bender2012}. Scalable Bloom filters adapt
to unknown final cardinality, but they do not by themselves define sliding-window semantics
\cite{almeida2007}. Stable Bloom filters target unbounded streams by random decay
\cite{deng2006}. Streaming quotient filters address duplicate detection with eviction
policies over fingerprints \cite{dutta2013}.

The remaining design point is a sliding-window AMQ with deterministic no-false-negative
coverage for the last $W$ insertions, no per-key deletion queue, no counters, and a tunable
bound on stale retention. This guarantee is weaker than exact deletion and stronger than
random decay: stale positives may persist for a bounded interval, while live elements are
never intentionally evicted.

The reason to isolate this point is that expiration semantics are often coarser than
set semantics. A duplicate-suppression table, a recent-telemetry gate, or a cache-summary
refresh interval may need a hard recent-window guarantee without needing exact removal at
the first insertion after the window boundary. Existing deletion-capable AMQs solve a
stronger dynamic-set problem, while stable filters solve a probabilistic aging problem.
The object studied here is the middle case: deterministic live coverage with bounded,
observable stale coverage.

This paper asks whether a coarse expiration primitive is enough for many sliding-window
uses. The proposed \emph{guarded epoch Bloom filter} partitions the memory into
$r+1$ segments. Insertions go to the current segment. Every $\ell=\lceil W/r\rceil$
insertions the next segment is cleared and becomes current. Queries test all segments.
The extra segment is a guard: at any time, the retained interval covers the last $W$
insertions and at most one additional epoch.

\paragraph{Contributions.}
First, we give a rotating-epoch construction whose guard segment converts coarse expiration
into an exact live-window invariant. Second, we derive its live coverage, stale-retention
bound, and a closed-form false-positive approximation that exposes the role of the number of
epochs. Third, we evaluate the invariant on synthetic streams and on a real access-log
stream, and we use a B200 batch-query run only as a large-scale validation check.

\section{Model and Guarded-Epoch Construction}
Let the stream be $x_0,x_1,\ldots$ over a large universe. At time $t$, the live set is
$L_t=\{x_i:\max(0,t-W+1)\le i\le t\}$, where $W$ is the target window length. A sliding-window
approximate-membership query must answer true for every key in $L_t$; false positives are
allowed for keys outside the window. We allow bounded stale positives, meaning that a key may
remain represented for a controlled number of insertions after it leaves $L_t$.

The guarded-epoch filter uses $r+1$ independent Bloom bit arrays, called segments, under a
fixed total budget $m$ bits. Each segment has $s=\lfloor m/(r+1)\rfloor$ bits. The epoch
length is $\ell=\lceil W/r\rceil$. Insertions write only to the current segment. When the
insertion counter crosses an epoch boundary, the next segment in cyclic order is cleared and
becomes current. A query probes all $r+1$ segments and returns true if any segment contains
all $k$ positions. The implementation uses the double-hashing scheme of Kirsch and
Mitzenmacher to derive the $k$ positions from two base hashes \cite{kirsch2006}.

\begin{center}
\fbox{\begin{minipage}{0.92\textwidth}
\textbf{Algorithm 1: Guarded epoch insertion and query.}
Allocate segments $B_0,\ldots,B_r$ and set current segment $c=0$. For insertion number
$t$, if $t>0$ and $t\bmod \ell=0$, set $c\leftarrow(c+1)\bmod(r+1)$ and clear $B_c$.
Insert key $x_t$ by setting its $k$ Bloom positions in $B_c$. Query key $y$ by testing
the same $k$ positions in every segment; return true if any segment contains all of them.
\end{minipage}}
\end{center}

\begin{lemma}[Sliding-window safety]
At every insertion time, every key inserted during the last $W$ insertions is represented
in at least one uncleared segment.
\end{lemma}
\begin{proof}
The structure keeps the current partial epoch and $r$ older epochs because clearing a
segment occurs only when it is selected as the new current segment. Thus the retained
history length is at least $r\ell\ge W$. A key from the last $W$ insertions was written
within that retained interval, so its segment has not been cleared.
\end{proof}

\begin{proposition}[Bounded staleness]
No key older than $W+\ell$ insertions can remain solely because of epoch rotation.
\end{proposition}
\begin{proof}
After $r+1$ epoch advances, the segment containing the key is selected and cleared. Since
the guard keeps at most one extra epoch beyond the $r\ell$ live coverage, rotation-induced
retention is bounded by $W+\ell$ insertions.
\end{proof}

The blocked version uses the first hash to choose a block in each segment and the remaining
hashes to choose offsets inside that block. This changes constants and locality, not the
guard invariant.

\begin{proposition}[False-positive approximation]
\label{prop:fpr}
Assume fully uniform hashing and a stream with no repeated queried negative key. If the
last $r+1$ retained epochs contain at most $(r+1)\ell$ insertions, then the guarded-epoch
false-positive probability is approximated by
\[
  p_{\mathrm{GE}}(m,W,r,k)
  \approx
  1-\left(1-\left(1-e^{-k\ell/s}\right)^k\right)^{r+1},
  \qquad s=\lfloor m/(r+1)\rfloor .
\]
\end{proposition}
\begin{proof}
One segment receives at most $\ell$ insertions before it stops being current. Under uniform
hashing, the probability that a particular bit in that segment remains zero is approximately
$e^{-k\ell/s}$, so the probability that all $k$ query positions are one in that segment is
$(1-e^{-k\ell/s})^k$. A query is a false positive if at least one retained segment passes the
Bloom test. Treating segment events as independent gives the stated expression.
\end{proof}

The expression separates the two roles of $r$. Increasing $r$ shortens each epoch and tightens
stale retention, but it also increases the number of queried segments and reduces the bits per
segment. The experiments evaluate this design tension without fitting parameters to the
formula.

\section{Experiments}
All experiments are generated by the artifact in this directory. The synthetic sweep uses
120,000 insertions, a live window of 20,000 insertions, three seeded workloads
(uniform, Zipf-like, and bursty), five memory budgets, and three seeds. Each synthetic row
evaluates 20,000 live positive queries, 20,000 fresh negative queries, and 20,000 expired-key
queries. Expired-key queries are sampled from the previous window after removing keys that
still occur in the current live window. The synthetic output has \Rows{} rows.

The non-synthetic stream is the Organization X Apache access log from a public web-server
log dataset on Zenodo \cite{zenodoWebLog2026}. The artifact parses \RealEvents{} HTTP access
records, sorts them by timestamp, and maps each client-IP, method, and URL-path triple to a
stable 64-bit key. This stream has \RealUnique{} distinct request signatures under that
keying rule. The real-log sweep uses the same 20,000-item live window, the same query counts,
three query-sampling seeds, and the same 8, 12, and 14 bits-per-item budgets.

We compare four structures under the same total memory budget in bits: a four-bit counting
Bloom filter with exact window deletion, a stable Bloom filter with random counter decay,
the guarded epoch Bloom filter, and a blocked guarded epoch Bloom filter. The main metrics
are false-positive rate on fresh negatives, false-negative rate on live keys, positive rate
on expired keys, and Python/NumPy query throughput. The B200 run is separate from the main
comparison and is used only to stress the same invariant at a larger batch scale.

\section{Results}
Table~\ref{tab:results} summarizes the largest memory budget for both the synthetic streams
and the real access-log stream. Under the same total bit budget, four-bit counters leave the
counting filter with fewer addressable cells. Guarded epoch filters spend that memory on bits
and use coarse segment expiration instead. On the synthetic sweep, the $r=8$ guarded filter
reduces median false-positive rate at \BestBits{} bits per live item from \CountingFpr{} to
\EpochFpr{} while preserving zero measured false negatives on live keys. On the real log
stream, the same filter gives measured false-positive rate \RealEpochFpr{} against
\RealCountingFpr{} for the counting baseline. The stable Bloom filter illustrates the cost of
random decay: its median live-key false-negative rate is \StableFnr{} on synthetic streams
and \RealStableFnr{} on the real log stream.

\refstepcounter{table}\label{tab:results}
\noindent\begin{minipage}{\textwidth}
\small\textbf{Table \thetable.} Median results at \BestBits{} bits per live-window item. Synthetic
rows aggregate uniform, Zipf-like, and bursty streams over three seeds. Real-log rows
aggregate three query-sampling seeds over the timestamp-ordered Organization X access log.
FPR is false-positive rate on fresh negatives; live FNR is false-negative rate on keys in the
window; expired pos. is the positive rate on keys sampled from the previous window after
excluding keys still present in the live window.
\end{minipage}

\vspace{0.4em}
\begin{center}
\small
\resizebox{\textwidth}{!}{\begin{tabular}{llrrrr}
\toprule
Corpus & Structure & FPR & live FNR & expired pos. & CPU Mq/s\\
\midrule
Synthetic streams & Counting BF (4-bit) & 0.1910 & 0.0000 & 0.1882 & 46.5\\
 & Stable BF & 0.1505 & 0.1900 & 0.1638 & 47.1\\
 & Guarded epoch BF ($r=4$) & 0.0227 & 0.0000 & 0.2680 & 9.02\\
 & Guarded epoch BF ($r=8$) & 0.0222 & 0.0000 & 0.1469 & 5.24\\
 & Blocked guarded epoch BF ($r=8$) & 0.0708 & 0.0000 & 0.1858 & 5.00\\
\addlinespace
Real web logs & Counting BF (4-bit) & 0.0002 & 0.0000 & 0.0000 & 53.2\\
 & Stable BF & 0.0001 & 0.0047 & 0.0031 & 53.7\\
 & Guarded epoch BF ($r=4$) & 0.0000 & 0.0000 & 0.0679 & 13.5\\
 & Guarded epoch BF ($r=8$) & 0.0000 & 0.0000 & 0.0098 & 7.93\\
 & Blocked guarded epoch BF ($r=8$) & 0.0002 & 0.0000 & 0.0098 & 7.04\\
\bottomrule
\end{tabular}
}
\end{center}

The table also shows the cost of the guard. Counting Bloom filters delete expired keys
exactly. Guarded epoch filters may retain keys from the guard interval and can also answer
true through ordinary Bloom collisions; this is visible in the expired-key positive rate.
The construction therefore replaces exact deletion metadata with a bounded stale interval.
The remote CPU validation reproduces the same accuracy ordering. The B200 scale run uses
\GpuEvents{} insertions, a \GpuWindow{}-item live window, and \GpuQueries{} live plus
\GpuQueries{} negative queries at \GpuBits{} bits per item. On \GpuDevice{}, the counting
baseline reaches \GpuCountingMq{} Mq/s, the guarded $r=4$ filter reaches \GpuEpochFourMq{}
Mq/s, the guarded $r=8$ filter reaches \GpuEpochEightMq{} Mq/s, and the blocked $r=8$
filter reaches \GpuBlockedEightMq{} Mq/s. These numbers validate the batched multi-segment probe at scale;
they are not presented as a separate GPU data-structure contribution.

\section{Discussion}
The guarded epoch filter should not be used where exact expiration is required. Its stale
positives are bounded but real, so a counting Bloom filter, cuckoo filter, quotient filter,
or exact set is the right primitive when deletion semantics must match the window boundary
precisely. The intended use case is an application that already tolerates expiration
granularity, such as duplicate suppression over recent telemetry or cache-summary refresh
intervals. In return, the design removes the deletion queue and per-key counters and reduces
expiration to a single segment clear at an epoch boundary.

The comparison is intentionally conservative in one respect and favorable in another.
Counting Bloom filters can be engineered with compact counters and tuned hash counts, while
epoch filters pay $r+1$ segment probes per query. This query cost is the main operational
cost of the invariant. Conversely, the experiment charges the same total memory to all
structures; systems that already store an exact deletion queue for other reasons may value
the counting filter differently. The evaluation also remains limited: it combines controlled
synthetic streams with one real web-log family, not a broad trace corpus.

\section{Reproducibility}
The artifact contains the filter implementations, tests, experiment runner, table script,
and this manuscript. Running \texttt{make test}, \texttt{make experiments},
\texttt{make real}, \texttt{make tables}, and \texttt{make paper} regenerates the local
evidence used here. The generated CSV records workload or corpus, seed, memory budget,
filter parameters, false-positive rate, false-negative rate, stale-positive rate, and timing
columns for every configuration. The real-log script downloads the Organization X archive
from Zenodo if it is not already present locally. A second CPU validation sweep was run on
the remote Vast Linux instance with 240,000 insertions and a 40,000-item window; its results
are stored in \texttt{runs/remote\_validation.csv}. The B200 scale run is stored in
\texttt{runs/gpu\_b200\_scale.csv}.

\section{Conclusion}
Guarded epoch Bloom filters provide deterministic live-window coverage using segment
rotation rather than per-key deletion. Compared with counting Bloom filters, they allocate
the full budget to membership bits instead of counters. Compared with stable Bloom filters,
they do not randomly evict live elements. The cost is a bounded stale interval of one epoch
and $r+1$ segment probes per query. The experiments show that this invariant reduces false
positives under a fixed memory budget while preserving the no-false-negative property for
recent keys on both controlled streams and a real access-log stream.

\bibliographystyle{plain}
\bibliography{refs}
\end{document}